\documentclass[11pt,titlepage]{article}
\usepackage{amsthm}

\newtheorem{proposition}{Proposition}
\newtheorem{definition}{Definition}
\newtheorem{theorem}{Theorem}

\usepackage[latin1]{inputenc}
\usepackage{rotating, rotfloat}
\usepackage[round]{natbib} 
\usepackage{graphics, graphicx, hyperref, url, amsmath, subfigure, amsthm, amsfonts, enumerate, epstopdf, booktabs, bm}
\hypersetup{
    pdfstartview={FitH},    
    colorlinks=true,        
    linkcolor=blue,         
    citecolor=blue,         
    filecolor=blue,         
    urlcolor=blue           
}
\def\u{\textbf{\textit{u}}} \def\x{\textbf{\textit{x}}}  \def\vv{\textbf{\textit{v}}}
\def\y{\textbf{\textit{y}}}  \def\b{\textbf{\textit{b}}} \def\c{\textbf{\textit{c}}}
\def\d{\textbf{\textit{d}}}  \def\0{\boldsymbol 0}       
  \def\t{\textbf{\textit{t}}}

\def\etal{\itshape et al.}

\newcommand{\D}{\mbox{\bf D}}

\newcommand{\HD}{{\mbox{\bf HD}}}

\newcommand{\PD}{{\mbox{\bf PD}}}
\newcommand{\MD}{{\mbox{\bf MD}}}
\newcommand{\ZD}{{\mbox{\bf ZD}}}

\newcommand{\EHD}{{\mbox{\bf EHD}}}
\newcommand{\EZD}{{\mbox{\bf EZD}}}
\newcommand{\bec}{\begin{center}}
\newcommand{\enc}{\end{center}}
\newcommand{\bee}{\begin{eqnarray*}}
\newcommand{\ene}{\end{eqnarray*}}
\newcommand{\beq}{\begin{equation}}
\newcommand{\eeq}{\end{equation}}

\begin{document}

\title{\bf On similarity of the sample depth contours}
\vskip 5mm

\author {{Xiaohui Liu$^{a, b}$   \footnote{Corresponding author's email: csuliuxh912@gmail.com.}
        }\\ \\[1ex]
        {\em\footnotesize $^a$ School of Statistics, Jiangxi University of Finance and Economics, Nanchang, Jiangxi 330013, China}\\
        {\em\footnotesize $^b$ Research Center of Applied Statistics, Jiangxi University of Finance and Economics, Nanchang,}\\ {\em\footnotesize Jiangxi 330013, China}\\
}

\maketitle

\begin{center}
{\sc Abstract}
\end{center}

In this paper, we investigate the similarity property of the sample projection depth contours. It turns out that some of these contours are of \emph{the same shape} with different sizes, following a similar fashion to the Mahalanobis depth contours. One advantage of this investigation is the potential of bringing convenience to the computation of the projection depth contours; the other one is that we may utilize this idea to extend both the halfspace depth and zonoid depth to versions that do not vanish outside the convex hull of the data cloud, aiming at overcoming the so-called `outside problem'. Examples are also provided to illustrate the main results.
\vspace{2mm}

\vspace{2mm}


\setlength{\baselineskip}{1.5\baselineskip}

\vskip 0.1 in
\section{Introduction}
\paragraph{}
\vskip 0.1 in \label{Introduction}

To facilitate constructing robust affine equivariant estimators/inferential procedures for multivariate observations, it is necessary first to generalize the natural linear ordering existing in univariate data into higher dimensional spaces. To this end, \cite{Tuk1975} suggested a useful tool named halfspace depth. One of its major advantage is its capability to induce a center-outward ordering for observations, and hence can be used in various applications, severing as the `multivariate order statistics'.

The idea behind this depth is quite heuristic. It motivates many other similar ordering tools, following different principles, nevertheless. Among them, the most famous ones include the Oja depth \citep{Oja1983}, simplicial depth \citep{Liu1990}, zonoid depth \citep{KM1997}, Mahalanobis depth \cite{ZS2000}, and projection depth \citep{Zuo2003}, etc. For convenience of preferring  one such function over another among different depth notions, \cite{ZS2000} proposed an axiomatic definition for the general notions of statistical depth function $\D(\x, P): \mathbb{R}^{d} \rightarrow [0, 1]$, where $P$ denotes the probability measure. For given $P$, an ideal depth function is expected to satisfy generally four properties, i.e., (a) \emph{affine-invariance}, (b) \emph{maximality at a center point}, (c) \emph{monotonicity related to the center point}, and (d) \emph{vanishing at infinity}.

Based on the statistical depth functions, it is quite convenient to induce some trimmed depth regions:
\begin{eqnarray*}
  \mathcal{R}(\tau) := \{\x \in \mathbb{R}^d: \D(\x, P) \ge \tau\}, \quad \text{ for given } \tau \in (0, 1].
\end{eqnarray*}
Hereafter, we restrict $\tau > 0$ to ensure the boundedness of $\mathcal{R}(\tau)$. The boundary of $\mathcal{R}(\tau)$ is usually referred to as the $\tau$-th depth contour. Depth contours are useful graphical tools, and usually utilized in practice for visual purposes \citep{RRT1999}. Serving as the generalized quantiles, they can also be used to construct some depth-based bootstrap regions for multivariate estimators in spaces with dimension $d > 1$ \citep{YehS97, WL2012}.

Unfortunately, the sample depth contours are usually computationally intensive. Besides the Mahalanobis depth, great effort is needed to compute a depth contour induced from depth functions; see, e.g., \cite{PM2010a, LMM2017} for Tukey's halfspace depth contours, \cite{MLB2009} for the zonoid depth contours, and \cite{LZW2013, LZ2015} for the projection depth contours, among others. Since the sample simplicial depth contours are \emph{not} convex, the related computation is even more complex. It seems that no trivial algorithm exists for exactly computing the simplicial depth contours currently.

\begin{figure}[H]
\centering
	\includegraphics[angle=0,width=4in]{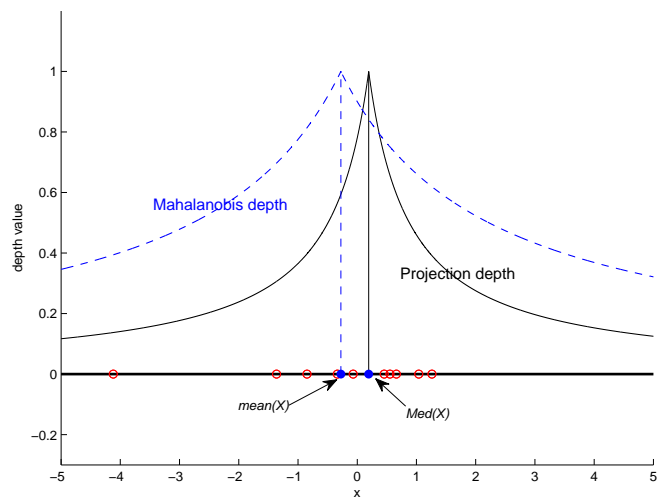}
\caption{Shown are the one dimensional projection depth and Mahalanobis depth function, where the hollow points stand for the observations, and the solid points denote the related sample mean and median, respectively.}
\label{fig:OneDimDepth}
\end{figure}

Observe that the definition way of the projection depth is very similar to that of the Mahalanobis depth in one dimensional space. Hence, their induced intervals are similar in the sense that they are symmetrical about the `center', i.e., the sample median for projection depth, and the sample mean for the Mahalanobis depth, respectively; see Figure~\ref{fig:OneDimDepth} for an illustration. Since the Mahalanobis depth contours in higher dimensional spaces inherit the similarity property because they are elliptic, it is then natural to wonder that whether or not the sample projection depth contours also enjoy this property in high dimensions. The similarity property is desirable because once it is already known that the contours are similar, we may employ it to facilitate the computation, especially when more than one similar contours are being computed simultaneously.

Usually, the depth regions/contours induced from the halfspace depth, simplicial depth and zonoid depth are not similar between each other, because they are completely determined by the data set \citep{Mos2013}. Hence, it is impossible to benefit their computation from the similarity property. On the other hand, we find in the sequel that the projection depth still possesses partly this property in spaces with $d > 1$, because it can induce some contours that are similar to each other when their depth values are small.

An another important usage of statistical depth functions is constructing depth-based classifiers \citep{GC2005}. Many classification methods based on statistical depth functions aforementioned have been developed from different principles during last decades; see \cite{DSG2016} and references therein for details. These depth-based classifiers usually enjoy many desirable properties. For example, most of them are affine invariant and may be robust against quite a proportion of outliers if a robust depth function, e.g., halfspace depth, is employed \citep{PV2015}.

However, some of them, such as the halfspace depth and zonoid depth, suffer from the so-called `outside problem'. That is, their depth values vanish outside the convex hull $\textbf{conv}(\mathcal{X}^n)$ of the given data cloud $\mathcal{X}^n = \{X_1, X_2, \cdots, X_n\} \subset \mathbb{R}^d$ ($d \ge 1$). This may bring great inconvenience to their practical applications in classification \citep{Hob2003, LMM2014}. How to overcome this is not trivial, especially for the case of Tukey's halfspace depth.

In this paper, we further consider this problem based on the discussions about the similarity property of the sample Mahalanobis depth and projection contours, and extend the conventional halfspace and zonoid depth to versions such that: (i) both of them coincide with their original counterparts inside in the convex hull of the data cloud, but (ii) do not vanish outside. These extensions are computable, and still satisfy all four properties of defining a general statistical depth function suggested by \cite{ZS2000}. They have nonsingular population versions, which is important in the practice of deriving the theoretical properties. Hence, we recommend to use them as an alterative to their original counterparts in applications such as classification.

The rest of this paper is organized as follows. We investigate the similarity property of the sample Mahalanobis depth and projection depth contours in Section \ref{Sec:MMS}. Based on this discussion, we propose the extended halfspace depth and extended zonoid depth in Section \ref{Sec:Extended}. Some illustrative examples are given in Section \ref{Sec:Illustrations}. Concluding remarks end this paper.

Throughout this paper, we mainly focus on the sample versions of the statistical depth functions and their associated contours having positive depth values, which are bounded, if no confusion arises.

\vskip 0.1 in
\section{Projection depth contours and similarity}
\paragraph{}
\vskip 0.1 in \label{Sec:MMS}

In the section, we will first investigate the similarity property existing partly in the sample projection depth contours, i.e.,
\begin{eqnarray*}
  \{\x \in \mathbb{R}^d: \PD(\x, P_n) \ge \tau\}, ~\forall \tau \in (0, 1].
\end{eqnarray*}
Following from \cite{Zuo2003}, here the projection depth is defined as follows:
\begin{eqnarray}
\label{eqn:PD}
  \PD(\x, P_n) = \frac{1}{1 + O(\x, P_n)},
\end{eqnarray}
where
\begin{eqnarray*}
  O(\x, P_n) = \sup_{\u \in \mathcal{S}^{d-1}} \frac{|\u^\top \x - \text{Med}(\u^\top \mathcal{X}^n)|}{\text{MAD}(\u^\top \mathcal{X}^n)},
\end{eqnarray*}
where $P_n$ denotes the empirical probability measure related to $\mathcal{X}^n$, $\mathcal{S}^{d-1} = \{\x: \|\x\| = 1, ~\x \in \mathbb{R}^d\}$, and $\u^{\top}\x$ denotes the projection of $\x$ onto the unit vector $\u$, and $\u^{\top}\mathcal{X}^{n} = \left\{\u^{\top}X_{1},\, \u^{\top}X_{2},\, \cdots,\, \u^{\top}X_{n}\right\}$. Let $Z_{(1)} \leq Z_{(2)} \leq \cdots \leq Z_{(n)}$ be the order statistics based on the univariate random variables $\mathcal{Z}^{n} = \{Z_{1},\, Z_{2},\, \cdots,\, Z_{n}\}$, then
\begin{eqnarray*}
  \mbox{Med}(\mathcal{Z}^{n}) &=& \frac{Z_{(\lfloor (n + 1) / 2\rfloor)} + Z_{(\lfloor (n + 2) / 2\rfloor)}}{2},\\
  \mbox{MAD}(\mathcal{Z}^{n}) &=& \mbox{Med}\{|Z_{i} - \mbox{Med}(\mathcal{Z}^{n})|,\, i = 1,\, 2,\, \cdots,\, n\},
\end{eqnarray*}
where $\lfloor \cdot \rfloor$ is the floor function.

For convenience, we assume that the given data cloud $\mathcal{X}^n$ are in general position throughout this paper. That is, there are no than $d$ data points in a $(d-1)$-dimensional hyperplane. This assumption is commonly imposed in the literature related to depth functions \citep{Don1982, MLB2009}. When $\mathcal{X}^n$ are in general position, it is easy to check that MAD($\u^\top \mathcal{X}^n) > 0$ for any $\u \in \mathcal{S}^{d-1}$.

Observe that the definition fashion of projection depth is similar to that of the Mahalanobis depth \citep{ZS2000}, i.e.,
\begin{eqnarray}
\label{eqn:MD}
  \MD(\x, P_n) = \frac{1}{1 + \sqrt{(\x - \bar{X})^\top \hat{\Sigma}^{-1} (\x - \bar{X})}},
\end{eqnarray}
where $\hat{\Sigma} = \frac{1}{n} \sum\limits_{i=1}^n (X_i - \bar{X}) (X_i - \bar{X})^\top$. Slightly different, the projection depth depends on the outlyingness function $O(\x, P_n)$, which is able to measure the outlyingness of $\x$ with respect to $\mathcal{X}^n$, through using the technique of projection pursuit, while the Mahalanobis depth is defined on the Mahalanobis distance of $\x$ to the sample mean $\bar{X}$. Since the sample Mahalanobis depth contours are of elliptical shape and in turn are similar, it is natural to expect that the sample projection depth contours also enjoy the same similarity property.

Formally, let's provide the definition of the similarity between two sample contours induced from a statistical depth function as follows.

\begin{definition}
\label{def:similarity}
  Let $\mathcal{C}_1$, $\mathcal{C}_2 \subset \mathbb{R}^d$ be two depth-induced contours with depth values $d_1$ and $d_2$, respectively. We say $\mathcal{C}_1$ and $\mathcal{C}_2$ to be similar if there exists a given point $\x_0 \in \mathbb{R}^d$ and a known function $h(\cdot, \cdot)$, conditionally on the given data set, such that: for any $\u \in \mathcal{S}^{d-1}$, we have
  \begin{eqnarray*}
    d_1 = h(\lambda_x, \u),\\
    d_2 = h(\lambda_y, \u),
  \end{eqnarray*}
  where $\x = \x_0 + \lambda_x \u$ and $\y = \x_0 + \lambda_y \u$ denote the intersection points of $\mathcal{C}_1$ and $\mathcal{C}_2$ for some $\lambda_x, \lambda_y > 0$ with the ray stemming from $\x_0$ along direction $\u$, respectively. Without loss of generality, we call $\x_0$ `similarity center', and $h(\cdot, \cdot)$ `generating function' of $\mathcal{C}_1$ and $\mathcal{C}_2$.
\end{definition}

\begin{figure}[H]
\centering
	\includegraphics[angle=0,width=4in]{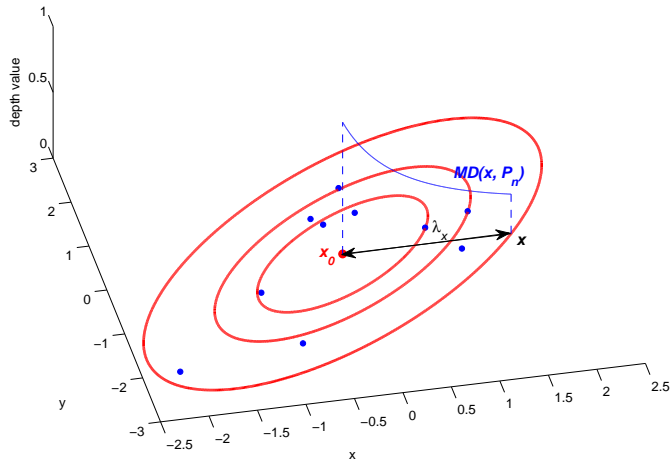}
\caption{Shown is the Mahalanobis depth function along a fixed ray stemming from $\x_0 =\bar{X}$, where the small points stand for the observations, and the big point is the sample mean.}
\label{fig:MDcontours}
\end{figure}

For any $\vv_0 \in \mathcal{S}^{d-1}$, let $\x = \bar{X} + \lambda_x \vv_0$. For the Mahalanobis depth given in \eqref{eqn:MD}, it is easy to check that the similarity center of its all sample contours is the sample mean $\bar{X}$, and the generating function satisfies
\begin{eqnarray*}
  h_M(\lambda_x, \vv_0) := \frac{1}{1 + \lambda_x \sqrt{\vv_0^\top \hat{\Sigma}^{-1} \vv_0}} = \MD(\x, P_n).
\end{eqnarray*}
That is, for the given data set, the Mahalanobis depth of $\x$ depends only on $d + 1$ unknowns, i.e., $\lambda_x$ and $\vv_0$, and is the inverse of a linear function with respect to the length $\lambda_x$ of $\x - \bar{X}$ once $\vv_0$ is given. Hence, $\MD(\x, P_n)$ decreases in a similar way when $\x$ is moving away from $\bar{X}$ along the same ray stemming from $\bar{X}$, and the related sample contours appears to be similar to each other; see Figure~\ref{fig:MDcontours} for an illustration.

Now, let's proceed to discuss the similarity property existing in the sample projection depth contours. In the literature, for given $\mathcal{X}^n$ in general position, \cite{LZW2013} have showed that there exist a finite number of data dependent direction vectors $\{\u_k\}_{k=1}^N$ such that
\begin{eqnarray}
    \label{eqn:PDsample}
    O(\x, P_n) = \max_{1\le k\le N} \frac{\u_k^\top \x - \u_k^\top \c_k}{\u_k^\top \d_k},
\end{eqnarray}
where $\c_{k}$'s and $\d_{k}$'s are some data dependent vectors, $N$ denotes the number of them. These vectors are known once the data set is given.

Based on this result, it is possible to obtain the following interesting result. Remarkable, since in one dimensional space, the definition ways of both projection depth and Mahalanobis depth follow a similar fashion. Trivially, the projection depth induced intervals are similar. Hence, we only present the result of $d \ge 2$ here.

\begin{theorem}
\label{th:PD}
  When $d \ge 2$, for the projection depth defined in \eqref{eqn:PD}, once $\lambda_x = \|\x - \x_{0, P}\| \ge \ell$, we have
  \begin{eqnarray*}
    \PD(\x, P_n) = h_P(\lambda_x, \vv_x),
  \end{eqnarray*}
  where $\ell$ is specified in \eqref{eqn:ell}, $\x_{0, P}$ stands for the sample projection depth median, and
  \begin{eqnarray*}
    h_P(\lambda_x, \vv_x) = \sum_{\vv_x \in \mathcal{C}_{k_*}} \frac{1}{1 + a_{k_*} + \lambda_x \cdot b_{k_*}(\vv_x)},
  \end{eqnarray*}
  with $\mathcal{C}_{k_*}$, $a_{k_*}$ and $b_{k_*}(\cdot)$ being specified in \eqref{eqn:cone} and \eqref{eqn:akbk}, respectively.
\end{theorem}

\begin{proof}[Proof of Theorem \ref{th:PD}]
  By \cite{Zuo2013}, it is already known that the sample projection depth median is unique. Without confusion, write the projection depth median related to $\mathcal{X}^n$ as $\x_{0, P}$. Using this, for any $\x \neq \x_{0, P}$, \eqref{eqn:PDsample} can be further transformed into the following form:
  \begin{eqnarray*}
    O(\x, P_n) &=& \max_{1\le k\le N} \frac{\u_k^\top (\x_{0, P} + \lambda_x \vv_x) - \u_k^\top \c_k}{\u_k^\top \d_k}\\
    &=& \max_{1\le k\le N} \left\{\frac{\u_k^\top \x_{0, P} - \u_k^\top \c_k}{\u_k^\top \d_k} + \lambda_x \frac{\u_k^\top \vv_x}{\u_k^\top \d_k}\right\},
  \end{eqnarray*}
  where $\lambda_x = \|\x - \x_{0, P}\|$ and $\vv_x = (\x - \x_{0, P}) / \lambda_x$.

  For simplicity, hereafter, we denote
  \begin{eqnarray*}
    \lambda_{k^*} (\vv_x) = \inf\left\{\lambda: a_{k_*} + \lambda \cdot b_{k_*}(\vv_x) \ge a_{k} + \lambda \cdot b_{k}(\vv_x),~ \forall k \neq k_*\right\},
  \end{eqnarray*}
  where
  \begin{eqnarray}
  \label{eqn:akbk}
    a_k = \frac{\u_k^\top \x_{0, P} - \u_k^\top \c_k}{\u_k^\top \d_k}, \text{ and } b_k(\vv_x) = \frac{\u_k^\top \vv_x}{\u_k^\top \d_k}.
  \end{eqnarray}
  with $k_*$ satisfying $b_{k_*}(\vv_x) = \max\limits_{1\le k\le N} b_{k}(\vv_x)$, i.e., the maximum slope among all linear functions $a_{k} + \lambda \cdot b_{k}(\vv_x)$'s with respect to $\lambda$. For given $\vv_x$, it is easy to check that: (i) $\lambda_{k^*} (\vv_x)$ is finite, and (ii) for any $\y = \x_{0, P} + \lambda_y \vv_x$ with $\lambda_y \ge \lambda_{k^*} (\vv_x)$, we have
  \begin{eqnarray*}
    O(\y, P_n) =  a_{k_*} + \lambda_y \cdot b_{k_*}(\vv_x).
  \end{eqnarray*}

  Observe that
  \begin{eqnarray}
  \label{eqn:cone}
    \mathcal{C}_{k_*} := \left\{\t \in \mathbb{R}^d: \frac{\u_{k_*}^\top \t}{\u_{k_*}^\top \d_{k_*}} \ge \frac{\u_l^\top \t}{\u_l^\top \d_l}, ~\forall l \neq {k_*}\right\}
  \end{eqnarray}
  is a cone, and for any $\vv \in \mathcal{C}_{k_*} \cap \mathcal{S}^{d-1}$, we always have $b_{k_*}(\vv) = \max\limits_{1\le k\le N} b_{k}(\vv)$. Similar to the case of $\vv_x$, it is easy to check that $\lambda_{k^*} (\vv)$ is finite, and continuous with respect to $\vv \in \mathcal{C}_{k_*} \cap \mathcal{S}^{d-1}$. Since $\mathcal{C}_{k_*}$ is a close set, we claim that
  \begin{eqnarray*}
    \gamma_{k_*} = \sup_{\vv \in \mathcal{C}_{k_*} \cap \mathcal{S}^{d-1}} \lambda_{k^*} (\vv) < +\infty.
  \end{eqnarray*}
  By further noting that the number of such nonempty cones $\mathcal{C}_{k_*}$'s is at most $N$, and all them together form the whole space $\mathbb{R}^{d}$, we obtain
  \begin{eqnarray}
  \label{eqn:ell}
    \ell = \max_{1\le k_* \le N,~ \mathcal{C}_{k_*} \text{ is nonempty}} \gamma_{k_*} < +\infty.
  \end{eqnarray}
  Using this, the proof of this theorem follows immediately.
\end{proof}

Let
\begin{eqnarray*}
  \tau_* = \inf_{\x \in \{\y\in\mathbb{R}^d: \|\y - \x_{0, P}\| = \ell\}} \PD(\x, P_n).
\end{eqnarray*}
Using Theorem \ref{th:PD}, it is easy to show the following result.
\begin{theorem}
\label{th:similarity}
  For any $\tau_1 < \tau_2 < \tau_*$, we have that the sample projection contours having depth values $\tau_1$ and $\tau_2$, respectively, are similar to each other in the sense of Definition \ref{def:similarity}. Their similarity center is the projection depth median, and generating function is $h_P(\cdot, \cdot)$.
\end{theorem}

As shown in \cite{LZW2013, LZ2014, LZ2015}, the sample projection depth contours are of polyhedral shape \emph{due to the piecewise linear property of the sample outlyingness function}. Observe that for any $\vv_x \in \mathcal{S}^{d-1}$, it is contained by one and only one cone as defined in \eqref{eqn:cone}. Hence, as a byproduct of Theorem \ref{th:PD}, it is actually easy to show the following result.

\begin{proposition}
\label{prop:facets}
  For given $\mathcal{X}^n$, the number of the facets of the $\tau$-th sample projection depth contour is less than $N$ for any $0 < \tau < \sup\limits_{\x} \PD(\x, P_n)$.
\end{proposition}

Proposition \ref{prop:facets} may be utilized to evaluate the complexity order of the algorithm for computing the projection depth contours. In practice, the true number of facets of a sample projection depth contour is far smaller than $N$.

Summarily, Theorems \ref{th:PD}-\ref{th:similarity} state an interesting result that the sample projection depth function may induce two kinds of contours: (i) some of them may \emph{not} enjoy the similarity property, i.e., the inner-most sample contours taking depth values larger than $\tau_*$; (ii) the others having depth values no larger than $\tau_*$ are similar to each other with their similarity center being the projection depth median, behaving similar to the sample Mahalanobis depth contours; see Figures \ref{fig:PDfunction}-\ref{fig:PDcontours} for illustrations. Since the second kind of contours are similar and may be generated by using the same generating function $h_P(\cdot, \cdot)$, the results of Theorems \ref{th:PD}-\ref{th:similarity} may be helpful in practice when more than one similar sample projection depth contours are being computed.

\vskip 0.1 in
\section{Extensions of the halfspace depth and zonoid depth}
\paragraph{}
\vskip 0.1 in \label{Sec:Extended}

Since the sample projection depth contours having depth values no larger than $\tau_*$ are similar to each other, we may consider these outside most contours as the copies of the $\tau_*$-th contour. They are in fact extended by using their similarity center and generating function, and hence have the same shape but are of different sizes. Observe that a statistical depth function can be completely characterized by its depth contours, and there are some depth functions, e.g., halfspace depth and zonoid depth, vanish outside the convex hull \textbf{conv}($\mathcal{X}^n$). Hence, we may use a similar continuation technique to the case of the sample projection depth to extend these depth functions to versions that do not suffer from the outside problem. This is the main focus of this section.

We review the definition of the halfspace depth and zonoid depth as follows. According to \cite{Tuk1975}, for given data $\mathcal{X}^n$, the halfspace depth of a point $\x \in \mathbb{R}^d$ with respect to $\mathcal{X}^n$ is given by
\begin{eqnarray*}
  \HD(\x, P_n) = \inf_{\u \in \mathcal{S}^{d-1}} P_n(\u^\top X \le \u^\top \x).
\end{eqnarray*}
In the literature, it is known that $\HD(\x, P_n)$ is stepwise, and vanishes outside the convex hull $\textbf{conv}(\mathcal{X}^n)$ of $\mathcal{X}^n$. It centers at Tukey's halfspace median $\x_{0,H}$, which is conveniently defined to be the average of all points in the inner-most halfspace depth trimmed region, which is usually not a singleton in various situations \citep{Don1982, LLZ2017}.

Following from \cite{KM1997}, the zonoid depth of $\x$ with respect to $\mathcal{X}^n$ is defined as
\begin{eqnarray}
\label{Def:ZD}
  \ZD(\x, P_n) &=&
  \begin{cases}
    \sup\left\{ \alpha: \x = \sum\limits_{i=1}^n p_i X_i, ~\sum\limits_{i=1}^n p_i = 1, ~np_i \in [0, 1/\alpha], ~\forall i \right\}, & \x \in \textbf{conv}(\mathcal{X}^n)\\[3ex]
    0, & \x \notin \textbf{conv}(\mathcal{X}^n).
  \end{cases}
\end{eqnarray}
Different from Tukey's halfspace depth, this depth function maximizes at the sample mean $\bar{X}$, but it also vanishes outside the convex hull $\textbf{conv}(\mathcal{X}^n)$ of $\mathcal{X}^n$.

Since both $\HD(\x, P_n)$ and $\ZD(\x, P_n)$ are completely characterized by the empirical probability measure \citep{Mos2013}, their induced sample contours obviously do not enjoy the similarity property as the Mahalanobis depth contours.

To improve them to versions that can overcome the so-called outside problem, we may borrow the idea lying behind the sample projection depth contours. That is, we remains their contours inside in the convex hull unchanged, and prolong the boundary of $\textbf{conv}(\mathcal{X}^n)$ to the area outside $\textbf{conv}(\mathcal{X}^n)$ based some generating functions with respect to their similarity centers, respectively, as did in the case of the sample projection depth. The key point here is to find a proper generating function $h(\cdot, \cdot)$.

For fixed $\vv \in \mathcal{S}^{d-1}$, observe that for any $\x = \bar{X} + \lambda_x \vv$ and $\y = \bar{X} + \lambda_y \vv$, we have that
\begin{eqnarray}
\label{eqn:ContTechMD}
  \frac{\MD(\x, P_n)}{\MD(\y, P_n)} = \frac{1 + \lambda_y \sqrt{\vv^\top \hat{\Sigma}^{-1} \vv}}{1 + \lambda_x \sqrt{\vv^\top \hat{\Sigma}^{-1} \vv}},
\end{eqnarray}
which is actually a ratio of two linear functions for given $\vv$. Similarly, for the case of the projection depth, when $\PD(\x, P_n) \leq \tau_*$ and $\PD(\y, P_n) \leq \tau_*$ with $\x = \x_{0, P} + \lambda_x \vv$ and $\y = \x_{0, P} + \lambda_y \vv$, we also have
\begin{eqnarray}
\label{eqn:ContTech}
  \frac{\PD(\x, P_n)}{\PD(\y, P_n)} = \frac{1 + a_{k_*} + \lambda_y \cdot b_{k_*}(\vv)}{1 + a_{k_*} + \lambda_x \cdot b_{k_*}(\vv)},
\end{eqnarray}
for some $k_* \in \{1, 2, \cdots, N\}$, based on Theorem \ref{th:PD}. Here $a_{k_*}$ and $b_{k_*}(\vv)$ depend only on $\mathcal{X}^n$ and $\vv$. Since $a_{k_*}$ and $b_{k_*}(\vv)$ are fixed for given $\vv$ and can be treated as constants, \eqref{eqn:ContTech} is also a ratio of two linear functions.

Bearing this in mind, we extend the conventional halfspace depth and zonoid depth to versions that may take positive depth values even outside the convex hull of the data cloud by using the similar continuation technique to \eqref{eqn:ContTechMD} and \eqref{eqn:ContTech} as follows.

The \emph{extended halfspace depth}.
\begin{eqnarray}
\label{eqn:EHD}
  \EHD(\x, P_n) =
  \begin{cases}
    \HD(\x, P_n), & \text{ if } \x \in \textbf{conv}(\mathcal{X}^n)\\[3ex]
    \frac{\lambda_{\vv_x}}{\lambda_x}\frac{1}{n}, & \text{ if } \x \notin \textbf{conv}(\mathcal{X}^n),
  \end{cases}
\end{eqnarray}
where $\lambda_x = \|\x - \x_{0, H}\|$ and $\lambda_{v_x} = \|\x_c - \x_{0, H}\|$ with $\x_c$ being the intersection point between the boundary of $\textbf{conv}(\mathcal{X}^n)$ and the ray stemming from $\x_{0, H}$ and passing through $\x$. Here we assume that $\frac{\lambda_{\vv_x}}{\lambda_x} \rightarrow 1$ if $\lambda_{\vv_x}\rightarrow +\infty$.

The \emph{extended zonoid depth}.
\begin{eqnarray}
\label{eqn:EZD}
  \EZD(\x, P_n) =
  \begin{cases}
    \ZD(\x, P_n), & \text{ if } \x \in \textbf{conv}(\mathcal{X}^n)\\[3ex]
    \frac{\tilde{\lambda}_{\vv_x}}{\tilde{\lambda}_x}\frac{1}{n}, & \text{ if } \x \notin \textbf{conv}(\mathcal{X}^n),
  \end{cases}
\end{eqnarray}
where $\tilde{\lambda}_x = \|\x - \bar{X}\|$ and $\tilde{\lambda}_{v_x} = \|\tilde \x_c - \bar{X}\|$ with $\tilde \x_c$ being the intersection point between the boundary of $\textbf{conv}(\mathcal{X}^n)$ and the ray stemming from $\bar{X}$ and passing through $\x$. Similarly, $\frac{\tilde{\lambda}_{\vv_x}}{\tilde{\lambda}_x} \rightarrow 1$ if $\tilde{\lambda}_{\vv_x}\rightarrow +\infty$.

Similar to \eqref{eqn:ContTechMD} and \eqref{eqn:ContTech}, it is easy to check that, for any $\x \notin \textbf{conv}(\mathcal{X}^n)$,
\begin{eqnarray*}
  \frac{\EHD(\x, P_n)}{\EHD(\x_c, P_n)} = \frac{\lambda_{v_x}}{\lambda_x}
\end{eqnarray*}
and
\begin{eqnarray*}
  \frac{\EZD(\x, P_n)}{\EZD(\tilde \x_c, P_n)} = \frac{\tilde{\lambda}_{v_x}}{\tilde{\lambda}_x}.
\end{eqnarray*}
That is, they are also ratios of two linear functions, \emph{which actually are the simplest linear functions among others}. Hence, these extended depth functions enjoy partly the similarity property as the projection depth function; see Figures \ref{fig:EHDidea}-\ref{fig:EZDidea} for illustrations.

Actually, there are some other depth functions, e.g., simplicial depth, also suffering from the so-called outside problem. Using the similar technique here, it is possible to extend them to versions that do not vanish outside the convex hull $\textbf{conv}(\mathcal{X}^n)$. We do not present these here for simplicity.

For the extended halfspace depth defined in \eqref{eqn:EHD} and the extended zonoid depth in \eqref{eqn:EZD}, the following theorem states that they still satisfy all four properties of defining a general statistical depth function \citep{ZS2000}. Hence, they can be used an alterative to their conventional counterparts in practical applications.

\begin{theorem}
\label{th:Extensions}
  Both the extended halfspace depth and the extended zonoid depth satisfy all four properties of defining a general statistical depth function.
\end{theorem}

\begin{proof}[Proof of Theorem \ref{th:Extensions}]
  We only show the case of the extended halfspace depth, because the proof for the extended zonoid depth is similar.

  Property (a) \emph{affine-invariance}. For any $\x \in \textbf{conv}(\mathcal{X}^n)$, the proof is trivial. When $\x \notin \textbf{conv}(\mathcal{X}^n)$, since both $\x_{0, H}$ and $\textbf{conv}(\mathcal{X}^n)$ are affine-invariant, it is easy to show that $\mathbf{A}\x_c + \b$ is still the intersection point of the ray, stemming from $\mathbf{A}\x_{0, H} + \b$ and passing through $\mathbf{A}\x + \b$, with $\textbf{conv}(\mathbf{A}\mathcal{X}^n + \b)$, where $\mathbf{A}$ denotes any $d\times d$ nonsingular matrix, $\b$ is a $d$-variate vector, and $\mathbf{A}\mathcal{X}^n + \b = \{\mathbf{A}X_1 + \b, \mathbf{A}X_2 + \b, \cdots, \mathbf{A}X_n + \b\}$. Hence, $\HD(\mathbf{A}\x_c + \b, P_{n, \mathbf{A}\mathcal{X}^n + \b}) = 1/n$. By noting that
  \begin{eqnarray*}
    \frac{\lambda_{v_{Ax+b}}}{\lambda_{Ax+b}} = \frac{\|(\mathbf{A}\x_c + \b) - (\mathbf{A}\x_{0, H} + \b)\|}{\|(\mathbf{A}\x + \b) - (\mathbf{A}\x_{0, H} + \b)\|} = \frac{\lambda_{v_x}\|\mathbf{A} \vv_x\|}{\lambda_{x}\|\mathbf{A} \vv_x\|} = \frac{\lambda_{v_x}}{\lambda_{x}},
  \end{eqnarray*}
  we have that $\EHD(\mathbf{A}\x + \b, P_{n, \mathbf{A}\mathcal{X}^n + \b}) = \EHD(\x, P_n)$, where $\vv_x = (\x - \x_{0, H}) / \lambda_x$, and $P_{n, \mathbf{A}\mathcal{X}^n + \b}$ denotes the empirical probability measure related to $P_{n, \mathbf{A}\mathcal{X}^n + \b}$.

  The proof of Property (b), i.e., \emph{maximality at a center point}, is trivial. For Property (c), \emph{monotonicity related to the center point}, by the construction of the extended halfspace depth, it is easy to check that its induced contours are still convex. Hence, Property (c) follows immediately.  For Property (d), i.e.,  \emph{vanishing at infinity}, when $\|\x\| \rightarrow +\infty$, we have $\lambda_x \rightarrow +\infty$, while $\x_c$ is always on the convex hull $\textbf{conv}(\mathcal{X}^n)$, which is bounded for any given data set $\mathcal{X}^n$. In this sense, $\lambda_{v_x} / \lambda_x \rightarrow 0$, as $\|\x\| \rightarrow +\infty$.

  This completes the proof of this theorem.
\end{proof}

In some occasions, we may need to derive the theoretical property of the statistical procedures relate these depth functions. A good depth function is expected to have a nonsingular population version. For the extended halfspace depth and the extended zonoid depth proposed above, we have the following result.

\begin{theorem}
\label{th:Convergence}
  Suppose $X_1, X_2, \cdots, X_n$ are i.i.d. copies of $X$. We have:
  \begin{enumerate}
    \item[(i)] The extended halfspace depth $\EHD(\x, P_n)$ converges in probability to the same population $\HD(\x, P)$ as that of $\HD(\x, P_n)$, for any given $\x \in \mathbb{R}^d$.
    \item[(ii)] When $E(\|X\|) < +\infty$, the extended zonoid depth $\EZD(\x, P_n)$ converges in probability to the same population $\ZD(\x, P)$ as that of $\ZD(\x, P_n)$, for any given $\x \in \mathbb{R}^d$.
  \end{enumerate}
\end{theorem}

The proof of Theorem \ref{th:Convergence} is trivial by noting that if $\x \notin \mathbf{conv}(\mathcal{X}^n)$, $\frac{\lambda_{v_x}}{\lambda_x} \frac{1}{n} \le \frac{1}{n} \rightarrow 0$ as $n \rightarrow \infty$. The rest of the proof follows immediately by using similar proofs to \cite{KM1997, Zuo2003}, respectively.

\vskip 0.1 in
\section{Illustrations}
\paragraph{}
\vskip 0.1 in \label{Sec:Illustrations}

In this section, we will use a real data example to illustrate the main results of this paper. The data set here is actually a part of the Boston housing data, which can be downloaded from \url{http://lib.stat.cmu.edu/datasets/boston}, or from the \textbf{Matlab} package accompanying with \cite{LZ2015}. Here we take onely the first 65 items of variables \textbf{rm}, \textbf{dis} for purpose of illustrations, where \textbf{rm} denotes the average number of rooms per dwelling, and \textbf{dis} the weighted distances to five Boston employment centres, respectively.

\begin{figure}[H]
\centering
	\includegraphics[angle=0,width=4in]{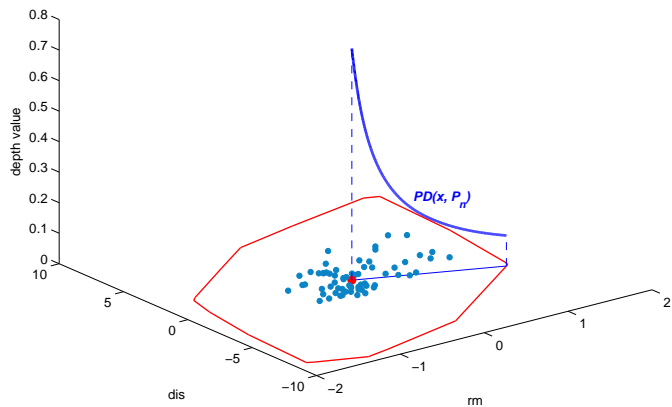}
\caption{Shown is the projection depth function along a fixed ray stemming from the projection depth median $\x_{0, P}$, where the small points stand for the observations, and the big point is the projection depth median.}
\label{fig:PDfunction}
\end{figure}

\begin{figure}[H]
\centering
	\includegraphics[angle=0,width=4in]{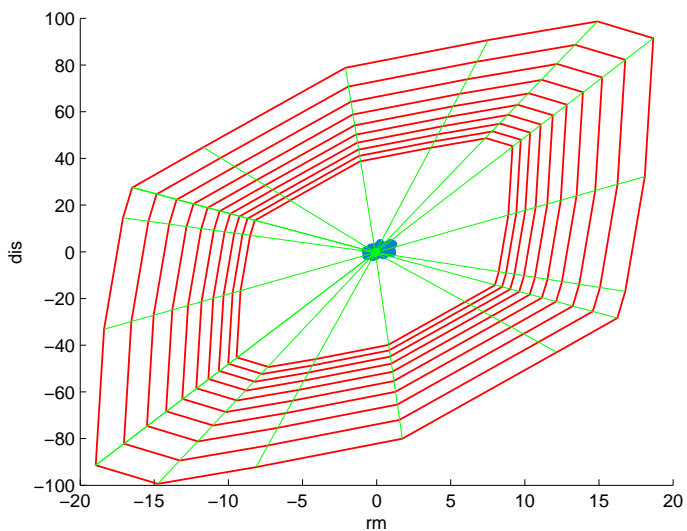}
\caption{Shown are ten sample projection depth contours with depth values less than 0.02, and the rays stemming from the projection depth median and passing through the vertices of these contours. The points in the center stand for the observations and their related projection depth median.}
\label{fig:PDcontours}
\end{figure}

Using this data set, we plot the projection depth function along the direction vector $\vv = (0.5022, -0.8648)^\top$ for an illustration. Figure \ref{fig:PDfunction} indicates that the projection depth value $\PD(\x, P_n)$ decreases very regularly when $\x$ is moving away from the projection depth median along $\vv$. Since for any $\vv$, $\PD(\x, P_n)$ decreases following a similar fashion, the projection depth contours with depth values small enough are similar to each other. As shown in Figure~\ref{fig:PDcontours}, the vertices of these contours lie in some rays stemming from the projection depth median. This confirms the theoretical results given in Section \ref{Sec:MMS}.

Furthermore, we also illustrate the idea behind the extended halfspace depth and the extended zonoid depth. As shown in Figures \ref{fig:EHDidea}-\ref{fig:EZDidea}, the contours outside the convex/having depth values smaller then $1/n$ are similar to the boundary of the convex hull with their similarity center to be the halfspace depth median and the sample mean, respectively. Different from their original counterparts, these two depth functions do not vanish outside the convex hull of the data set.

\begin{figure}[H]
\centering
	\includegraphics[angle=0,width=4in]{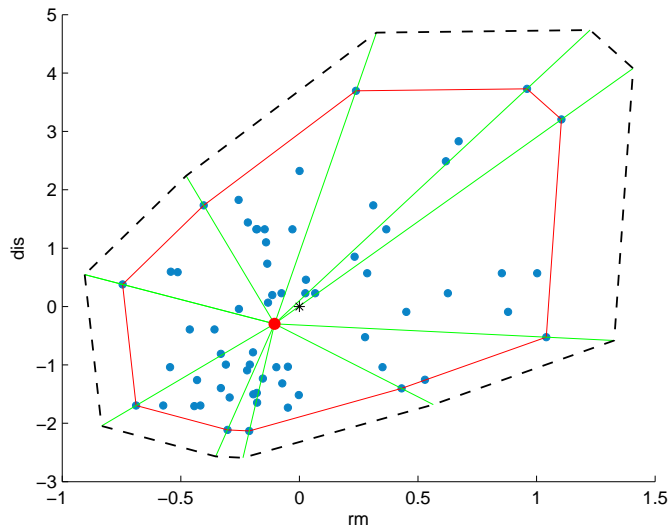}
\caption{Shown are two contours of the extended halfspace depth with depth values 0.8/65, 1/65 from the periphery inwards, and rays stemming from the halfspace depth median and passing through the vertices of $\textbf{conv}(\mathcal{X}^n)$. The star stands for the sample mean, the big point is the halfspace depth median, and the small points are the observations.}
\label{fig:EHDidea}
\end{figure}

\begin{figure}[H]
\centering
	\includegraphics[angle=0,width=4in]{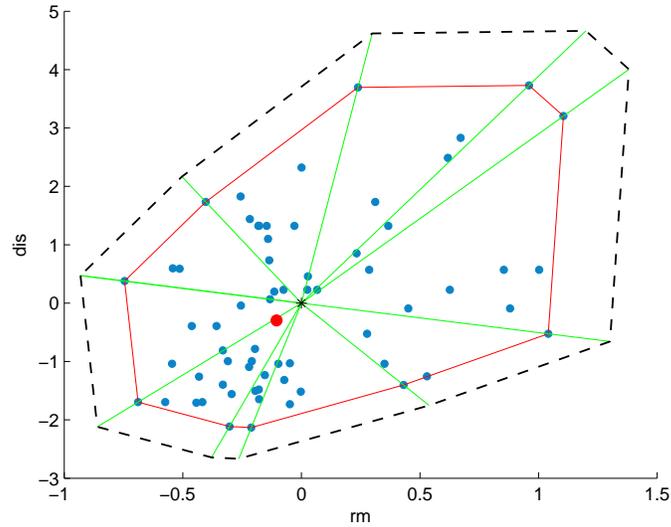}
\caption{Shown are two contours of the extended zonoid depth with depth values 0.8/65, 1/65 from the periphery inwards, and rays stemming from the sample mean and passing through the vertices of $\textbf{conv}(\mathcal{X}^n)$.}
\label{fig:EZDidea}
\end{figure}

Figures \ref{fig:EHC}-\ref{fig:EZC} shows the depth contours with depth values being 0.0092, 0.0123, 0.4462, 0.3846, 0.3231, 0.2615, 0.2000, 0.1385, 0.0769, 0.0154 for the extended halfspace depth, and 0.0092, 0.0123, 0.0154, 0.1560, 0.2967, 0.4374, 0.5780, 0.7187, 0.8593, 1.0000 for the extended zonoid depth, respectively, from the periphery inwards. The innermost contours are the same as the original halfspace depth and zonoid depth, respectively.

\begin{figure}[H]
\centering
	\includegraphics[angle=0,width=4in]{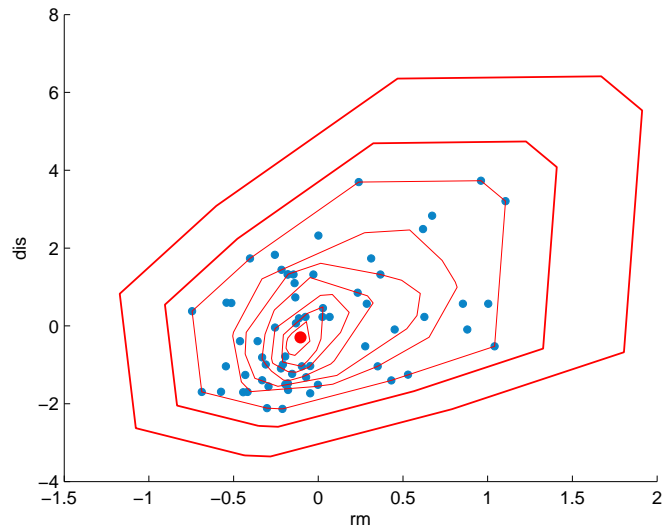}
\caption{Show are ten contours of the extended halfspace depth centering at the halfspace median.}
\label{fig:EHC}
\end{figure}

\begin{figure}[H]
\centering
	\includegraphics[angle=0,width=4in]{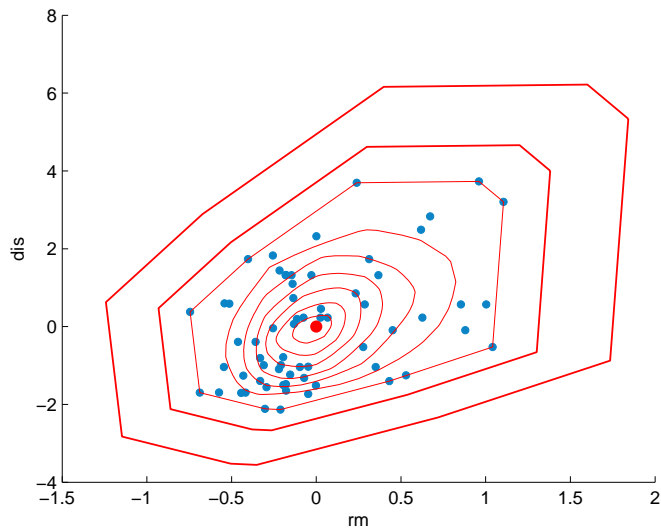}
\caption{Show are ten contours of the extended zonoid depth centering at the sample mean.}
\label{fig:EZC}
\end{figure}

\vskip 0.1 in
\section{Concluding Remarks}
\paragraph{}
\vskip 0.1 in \label{Sec:Remarks}

In this paper, we first considered the similarity property existing in the Mahalanobis depth and projection depth. It turns out that the projection depth behaves partly similar the Mahalanobis depth in the sense that it can induces some contours that are similar to each other. We then deeply explored the idea behind the similarity, and employed this idea to extend two well-known depths, i.e., halfspace depth and zonoid depth, to versions that do not vanish outside the convex hull of the data set.

It is worth mentioning that since the convex hull of the data set is of polyhedral shape, and their vertices should be some data points. Hence, their sample contours outside the convex hull is \emph{computable}, and the computation is in fact \emph{not very complex}, while for computing the ordinary sample contours inside in the convex hull, various mature algorithm have already been developed; see, e.g., \cite{PM2010a, LMM2017} and \cite{MLB2009} for halfspace depth and zonoid depth, respectively. For the case of the extended halfspace depth, we need to compute the halfspace depth median, which is computationally intensive. Fortunately, there is an exact algorithm for the halfspace depth median implemented by \textbf{C++} in the literature now; see e.g. \cite{LMM2017, LLZ2017}. Hence, the extended depths are still computable.

Furthermore, they satisfy all four properties of defining a general statistical depth function and have nonsingular populations. Hence, they are desirable alterative to the conventional halfspace depth and zonoid depth in applications, e.g., classification.

\vskip 0.1 in
\section*{Acknowledgements}
\paragraph{}
\vskip 0.1 in

The research is supported by NNSF of China (Grant No.11601197, 11461029), China Postdoctoral Science Foundation funded project (2016M600511, 2017T100475),  NSF of Jiangxi Province (No.20171ACB21030, 20161BAB201024), and the Key Science Fund Project of Jiangxi provincial education department (No.GJJ150439).


\medskip


\bigskip
\begin{thebibliography}{44}
    \footnotesize \vskip 0.1 in\setlength{\itemsep}{0.03in}
    \bibitem[Dutta \etal(2016)]{DSG2016} Dutta, S., Sarkar, S., Ghosh, A. K. (2016). Multi-scale classification using localized spatial depth. Journal of Machine Learning Research, 17(218), 1-30.
    \bibitem[Donoho(1982)]{Don1982} Donoho, D.L., 1982. Breakdown properties of multivariate location estimators. Ph.D. Qualifying Paper. Dept. Statistics, Harvard University.
    \bibitem[Ghosh and Chaudhuri(2005)]{GC2005} Ghosh, A. K., Chaudhuri, P. (2005). On maximum depth and related classifiers. Scandinavian Journal of Statistics, 32(2), 327-350.
    \bibitem[Hoberg(2003)]{Hob2003} Hoberg, R. (2003). Clusteranalyse, Klassifikation und Datentiefe; Reihe Quantitative \"Okonomie Band 129.
    \bibitem[Lange \etal(2014)]{LMM2014} Lange, T., Mosler, K., Mozharovskyi, P. (2014). Fast nonparametric classification based on data depth. Statistical Papers, 55(1), 49-69.
    \bibitem[Liu(1990)]{Liu1990} Liu, R. Y. (1990). On a notion of data depth based on random simplices. {\itshape Ann. Statist.}, 18: 191-219.
    \bibitem[Liu \etal(2014)]{LMM2017} Liu, X., Mosler, K., Mozharovskyi, P. (2014). Fast computation of Tukey trimmed regions in dimension $ p> 2$. arXiv preprint arXiv:1412.5122.
    \bibitem[Liu \etal(2017)]{LLZ2017} Liu, X., Luo, S., Zuo, Y. (2016). Some results on the computing of Tukey's halfspace medain. Statistical papers, https://doi.org/10.1007/s00362-017-0941-5.
    \bibitem[Liu and Zuo(2015)]{LZ2015} Liu, X., Zuo, Y. (2015). CompPD: A MATLAB package for computing projection depth. Journal of Statistical Software, 65(1), 1-21.
    \bibitem[Liu and Zuo(2014)]{LZ2014} Liu, X., Zuo, Y. (2014). Computing projection depth and its associated estimators. Statistics and Computing, 24(1), 51-63.
    \bibitem[Liu \etal(2013)]{LZW2013} Liu, X., Zuo, Y., Wang, Z. (2013). Exactly computing bivariate projection depth contours and median. Computational Statistics \& Data Analysis, 60, 1-11.
    \bibitem[Koshevoy and Mosler(1997)]{KM1997} Koshevoy, G., Mosler, K. (1997). Zonoid trimming for multivariate distributions. The Annals of Statistics, 1998-2017.
    \bibitem[Mosler(2013)]{Mos2013} Mosler, K. (2013). Depth statistics. \textit{In Robustness and complex data structures} (pp. 17-34). Springer Berlin Heidelberg.
    \bibitem[Mosler and Hoberg(2006)]{MH2006} Mosler, K., Hoberg, R. (2006). Data analysis and classification with the zonoid depth. DIMACS Series in Discrete Mathematics and Theoretical Computer Science, 72, 49.
    \bibitem[Mosler \etal(2009)]{MLB2009} Mosler, K., Lange, T., Bazovkin, P., (2009). Computing zonoid trimmed regions of dimension $d > 2$. Comput. Statist. Data Anal. 53, 2500-2510.
    \bibitem[Oja(1983)]{Oja1983} Oja, H. (1983). Descriptive statistics for multivariate distributions. Statistics \& Probability Letters, 1(6), 327-332.
    \bibitem[Paindaveine and \v{S}iman(2012)]{PM2010a} Paindaveine, D., \v{S}iman, M. (2012). Computing multiple-output regression quantile regions. Comput. Statist. Data Anal. 56, 840-853.
    \bibitem[Paindaveine and Van Bever(2015)]{PV2015} Paindaveine, D., Van Bever, G. (2015). Nonparametrically consistent depth-based classifiers. Bernoulli, 21(1), 62-82.
    \bibitem[Rousseeuw \etal(1999)]{RRT1999} Rousseeuw, P. J., Ruts, I., Tukey, J. W. (1999). The bagplot: a bivariate boxplot. The American Statistician, 53(4), 382-387.
    \bibitem[Tukey(1975)]{Tuk1975} Tukey, J.W. (1975). Mathematics and the picturing of data. {\itshape In Proceedings of the International Congress of Mathematicians}, 523-531. Cana. Math. Congress, Montreal.
    \bibitem[Wei and Lee(2012)]{WL2012} Wei, B., Lee, S. M. (2012). Second-order accuracy of depth-based bootstrap confidence regions. Journal of Multivariate Analysis, 105(1), 112-123.
    \bibitem[Yeh and Singh(1997)]{YehS97} Yeh, A., Singh, K. (1997). Balanced confidence regions based on Tukey's depth and the bootstrap. Journal of the Royal Statistical Society: Series B (Statistical Methodology), 59: 639-652.
    \bibitem[Zuo(2003)]{Zuo2003} Zuo, Y. (2003). Projection based depth functions and associated medians. The Annals of Statistics, 31: 1460-1490.
    \bibitem[Zuo(2013)]{Zuo2013} Zuo, Y. (2013). Multidimensional medians and uniqueness. Computational Statistics \& Data Analysis, 66, 82-88.
    \bibitem[Zuo and Serfling(2000)]{ZS2000} Zuo, Y., Serfling, R. (2000). General notions of statistical depth function. {\itshape Ann. Statist.}, 28: 461-482.
\end{thebibliography}
\end{document}